\documentclass[10pt]{article}
\usepackage{amsmath}
\usepackage{amssymb}
\numberwithin{equation}{section}
\usepackage{amsthm}
\usepackage{caption}
\usepackage{subfigure}
\usepackage{mathrsfs}
\usepackage{cite}
\theoremstyle{plain}
\newtheorem{proposition}{Proposition}[section]
\newtheorem{corollary}[proposition]{Corollary}
\newtheorem{lemma}[proposition]{Lemma}
\newtheorem{theorem}[proposition]{Theorem}
\theoremstyle{definition}
\newtheorem{definition}[proposition]{Definition}
\newtheorem{example}[proposition]{Example}
\newtheorem{remark}[proposition]{Remark}

\usepackage[dvips]{graphicx}
\usepackage[dvips]{color}
\usepackage{tikz}
\tikzset{
  source/.style={circle,draw=black!100,fill=black!50,inner sep = 0,minimum size=2mm},
  sink/.style={circle,draw=black!100,fill=white,inner sep = 0,minimum size=2mm}
}

\DeclareMathOperator{\sgn}{sgn}

\newcommand{\ds}{\displaystyle}
\newcommand{\R}{\mathbf{R}}
\newcommand{\C}{\mathbf{C}}
\newcommand{\Oh}{\mathcal{O}}
\newcommand{\smallfrac}[2]{{\textstyle\frac{#1}{#2}}}
\newcommand{\abs}[1]{\left\lvert #1 \right\rvert}
\newcommand{\avg}[1]{\bigl\langle #1 \bigr\rangle}

\begin{document}

\title{Peakon-antipeakon interactions in the Degasperis-Procesi Equation}
\author{
  Jacek Szmigielski\thanks{Department of Mathematics and Statistics, University of Saskatchewan, 106 Wiggins Road, Saskatoon, Saskatchewan, S7N 5E6, Canada; szmigiel@math.usask.ca}
  \and  Lingjun Zhou\thanks{Department of Mathematics, Tongji University, Shanghai, P.R. China; zhoulj@tongji.edu.cn}}

\date{\today}
\maketitle
\begin{abstract}
Peakons are singular, soliton-like solutions to nonlinear wave equations
whose dynamics can be studied using ordinary differential equations (ODEs).
The Degasperis-Procesi equation (DP) is an important example of an integrable PDE
exhibiting wave breaking in the peakon sector thus affording an interpretation
of wave breaking as a mechanical collision of particles.
In this paper we set up a general formalism in which to study collisions of DP peakons and apply it, as an illustration, to a detailed study of
three colliding peakons.  It is shown that peakons can collide only in pairs, no triple collisions are allowed and at the
collision a shockpeakon is created.
We also show that the initial configuration of peakon-antipeakon pairs is nontrivially correlated with the spectral properties of an accompanying non-selfadjoint boundary value problem.
In particular if peakons or antipeakons are bunched up on one side relative
to a remaining antipeakon or peakon then the spectrum is real and simple.  Even though the spectrum is in general complex the existence of a global solution in  either time direction dynamics is shown to imply the reality of the spectrum of the
boundary value problem.
\end{abstract}

\section{Introduction}
\vspace{1 cm}
The prototypical example of PDEs admitting peaked solitons is the family
\begin{equation}
  \label{eq:b-family}
  u_t - u_{xxt} + (b+1) u u_x = b u_x u_{xx} + u u_{xxx},
\end{equation}
often written as
\begin{equation}
  \label{eq:b-family-m}
  m_t + m_x u + b m u_x = 0, \qquad m = u - u_{xx},
\end{equation}
which was introduced by Degasperis, Holm and Hone \cite{degasperis-holm-hone},
and shown to be Hamiltonian for all values of $b$ \cite{holm-hone}.
The most studied cases are the Camassa--Holm (CH) equation ($b=2$),
and the
Degasperis--Procesi (DP) equation \cite{degasperis-procesi,degasperis-holm-hone}
($b=3$).
For $b>0$ these are the only values of $b$ for which the equation is integrable,
according to a variety of integrability tests
\cite{degasperis-procesi,mikhailov-novikov-perturbative,hone-wang-prolongation-algebras,ivanov-integrability-test}.
The case $b=0$ is relevant for another reason; this case provides a regularization
of the inviscid Burgers equation that is Hamiltonian and has classical
solutions globally in time \cite{bhat-bzero}.  The $b$-family admits an appealing
geometric interpretation
as Euler-Arnold equations on the space of densities $m(x)dx^b$ for the group of
orientation-preserving diffeormorphisms $\mathrm{Diff}(S^1)$ \cite{lenells-misiolek-tiglay}.

In order to discuss peakon solutions one needs to develop the concept
of weak solutions.  Because of the role that the Lax pair formalism plays in the theory we will define weak solutions in such a way that
the PDE in question is the
compatibility condition of weak Lax pairs.  This prompts the $x$ member of the Lax
pair to be viewed as an ODE with distribution coefficients while
the $t$ equation of the Lax pair is viewed as a (isospectral) deformation of the
former.  We subsequently need to rewrite the PDE itself as a distribution equation.  To this end we observe that the formulation \eqref{eq:b-family-m} suffers from the problem that
the product $mu_x$ is ill-defined already in the case of continuous, piecewise smooth
(in $x$) $u(x,t)$, since the quantity
$m=u-u_{xx}$ is a measure with a non-empty
singular support at the points of non-smoothness.  To make matters worse that measure is in addition
multiplied by the function $u_x$ which has jump discontinuities
exactly at those points.  This problem can be resolved easily; one instead rewrites \eqref{eq:b-family} as
\begin{equation}
  \label{eq:b-family-alt}
  (1-\partial_x^2) u_t
  + (b+1-\partial_x^2) \, \partial_x \left( \smallfrac12 \, u^2 \right)
  + \partial_x \left( \smallfrac{3-b}{2} \, u_x^2 \right)
  = 0.
\end{equation}
The case $b=3$ is of particular interest to us.  Then the term $u_x^2$ is absent from equation \eqref{eq:b-family-alt} and in that particular case one requires only that
$u(\cdot,t) \in L^2_{\mathrm{loc}}(\R)$; this means that the
DP equation can admit solutions $u$ that are not continuous
\cite{coclite-karlsen-DPwellposedness,coclite-karlsen-DPuniqueness,lundmark-shockpeakons}.

\emph{Multipeakons} are weak solutions of the form
\begin{equation}
  \label{eq:peakon-ansatz}
  u(x,t) = \sum_{i=1}^n m_i(t) \, e^{-\abs{x-x_i(t)}},
\end{equation}
formed through superposition of $n$ \emph{peakons}
(peaked solitons of the shape $e^{-\abs{x}}$).
This ansatz satisfies the PDE \eqref{eq:b-family-alt} if and only if the
positions $(x_1,\dots,x_n)$ and momenta $(m_1,\dots,m_n)$ of the
peakons obey the following system of $2n$ ODEs:
\begin{equation}
  \label{eq:b-peakon-ode}
  \dot x_k = \sum_{i=1}^n m_i \, e^{-\abs{x_k-x_i}},
  \qquad
  \dot m_k = (b-1) \, m_k \sum_{i=1}^n m_i \, \sgn(x_k-x_i) \, e^{-\abs{x_k-x_i}}.
\end{equation}
Here, $\sgn x$ denotes the signum function, which is $+1$, $-1$ or $0$
depending on whether $x$ is positive, negative or zero.
In shorthand notation,
with $\avg{f(x)}$ denoting the average of the left and right limits,
\begin{equation}
  \label{eq:average-notation}
  \avg{f(x)} = \frac12 \bigl( f(x^-)+f(x^+) \bigr),
\end{equation}
the ODEs can be written as
\begin{equation}
  \label{eq:b-ode-short}
  \dot x_k = u(x_k),
  \qquad
  \dot m_k = -(b-1) \, m_k \, \avg{u_x(x_k)}.
\end{equation}
In the CH case ($b=2$) this is a canonical Hamiltonian system
generated by $h=\frac12 \sum_{j,k=1}^n m_j \, m_k \, e^{-\abs{x_j-x_k}}$,
for which $x_j$s and $m_j$s are canonical positions and momenta.
In the DP case ($b=3$) this is a non-canonical Hamiltonian system
with the Hamiltonian $H=\sum_{j=1} ^n m_j$ and a non-canonical Poisson structure
given in \cite{degasperis-holm-hone2}.

It is important to distinguish the case of {\sl pure peakons} (initial $m_j(0)>0$)
or {\sl pure anti-peakons} (initial $m_j(0)<0$) from
a general case of {\sl multipeakons} (no restriction on the signs of $m_j(0)$).
Pure peakons have peaks,
pure anti-peakons have troughs while multipeakons contain both peaks and troughs.

The relevance of multipeakon solutions is that they provide a concrete model for wave breaking \cite{bss-peakon-antipeakon, bss-moment}.  For more information on the wave breaking
phenomenon for this class of wave equations the reader is referred to
\cite{McKean-breakdown} and \cite{constantin-escher, liu1, liu2, liu-escher}.
In the CH
case the distinction between pure peakons or anti-peakons and multipeakons does not result in a serious departure from the
inverse spectral formulas for pure peakons.
Indeed, explicit formulas for the $n$-peakon solution of the CH equation
were derived by Beals, Sattinger and Szmigielski \cite{bss-stieltjes} and then extended to $n$-multipeakons in \cite{bss-moment, beals-sattinger-szmigielski-toda}
using inverse spectral methods and the theory of orthogonal polynomials.  The situation for the DP equation
is considerably different.  The analysis of pure peakon solutions for the DP
equation
was accomplished by Lundmark and Szmigielski
\cite{lundmark-szmigielski-DPshort,ls-cubicstring} using inverse spectral methods and M.G. Krein's theory of oscillatory kernels
\cite{gantmacher-krein}.  In short, in these papers, it was shown that in the DP case, when working with pure peakon or pure anti-peakon solutions,
the concept of {\sl total positivity} plays a fundamental role, for example, implying that the spectrum involved is positive and simple.  For this reason going beyond
the pure peakon sector  of the DP will not be as straightforward as in the CH case which remains self-adjoint in the whole multipeakon sector.  The DP spectral problem, by contrast, is
manifestly non-selfadjoint.  Yet,
in addition to a general interest in modelling the wave breaking mechanism, there is another
reason for studying multipeakon solutions of the DP equation: in \cite{lundmark-shockpeakons}
Lundmark introduced  a new type of solution, {\sl  a shockpeakon solution}, which he showed for the case  $n=2$ gives a unique entropy weak solution originating from the peakon-antipeakon solution. Therefore, multipeakon solutions
can also provide us with an additional  insight into the onset of shocks.

The paper is organized as follows.  In Section \ref{sec:Fundamentals}, we set up
the formalism for an arbitrary number $n$ of multipeakons and elaborate on
general forms of peakons with index $1$ and $n$.  We emphasize the role of the
boundary value problem and its adjoint, both associated with the $x$-member of the Lax operator.
In Section \ref{sec:3peakons}, we undertake a detailed study of three multipeakons.
We establish the analytic character of solutions in Lemma \ref{lem:analm1m3} and describe the main properties of colliding pairs, culminating in Theorem \ref{thm:shockcreation} describing the creation of a shockpeakon at the collision.  Section \ref{sec:3peakspectral}
is devoted to analysis of the spectrum of the boundary value problem, in
particular we prove a signature-type Lemma \ref{lem:signeigen} relating the signs
of masses $m_j(0)$ of colliding peakons to the real parts of eigenvalues. Finally,
in Section \ref{sec:3peakclass} we classify different asymptotic in $t$ behaviour
of three multipeakons in terms of the sign configurations of the initial masses $m_j(0)$.

\vspace{1 cm}
\section{Lax pair and the multipeakon spectral problem}\label{sec:Fundamentals}
It was shown in \cite{degasperis-holm-hone} that the DP equation admits
the Lax pair:
\begin{equation}\label{eq:Lax}
(\partial_x-\partial_{xxx})\Psi=zm\Psi, \qquad
\Psi_t=[z^{-1}(1-\partial^2_x)+u_x-u\partial_x]\Psi.
\end{equation}
In particular if $u$ is given by the multipeakon ansatz \eqref{eq:peakon-ansatz},
\begin{equation}\label{eq:m}
m=2\sum_{i=1}^n m_i \delta_{x_i},
\end{equation}
and equations \eqref{eq:b-ode-short} for $b=3$ follow from the (distributional) compatibility of equations \eqref{eq:Lax}.
The boundary conditions consistent with the asymptotic behaviour of $\Psi$
read:
\begin{equation}\label{eq:asymptotics}
\Psi\sim e^x,  \text{ as } x\rightarrow -\infty, \qquad  \Psi \text{ is bounded as   } x\rightarrow +\infty.
\end{equation}
To see how the implementation of these conditions leads to an isospectral
problem we will trace back the most important steps in analysis of Lax
pair for peakons.  For more details the reader is referred to \cite{ls-cubicstring}.  We start in the region $x<x_1$ lying outside of the support of the
discrete measure $m$.  There, the first equation in the Lax pair can easily be solved and the boundary condition implemented by
$\Psi(x)=e^x.$
When $x_k<x<x_{k+1}, $ we have
\begin{equation}\label{eq:AkBkCk}
\Psi(x)=A_k(z)e^x+B_k(z)+C_k(z)e^{-x}, \qquad 1\leq k\leq n.
\end{equation} The coefficients $A_k(z),B_k(z),C_k(z)$ are polynomials of degree $k$ in $z$ given by
{\small
\begin{equation} \label{eq:kABC} \left(\begin{array}{c}A_k(z)\\ B_k(z)\\ C_k(z) \end{array}\right)=\left(\begin{array}{c}1\\ 0\\ 0 \end{array}\right)+\sum_{p=1}^k\left[\sum_{I\in {{[1,k]}\choose p}}\left(\prod_{i\in I}m_i\right)\left(\prod_{j=1}^{p-1}(1-e^{x_{i_j}-x_{i_{j+1}}})^2\right)\left(\begin{array}{c}1\\ -2e^{x_{i_p}}\\e^{2x_{i_p}}  \end{array}\right) \right](-z)^p, \end{equation} }
  where $\binom{[1,n]}{p}$ is the set of all $p$-element subsets
  $I=\left\{ i_1 < \dots < i_p \right\}$ of $\left\{ 1,\dots,n \right\}$.

For $x>x_n$ we will drop the subscript $n$, thus
\begin{equation}\label{eq:ABC}
\Psi(x)=A(z)e^x+B(z)+C(z)e^{-x}.
\end{equation}

 In the case of interest for us, namely $n=3$, the coefficients $A,B,C$ can
 be written explicitely
 \begin{subequations}\label{eq:3ABC}
\begin{align}
&A(z)=1-\big[m_1+m_2+m_3\big]z\nonumber\\
&+\big[m_1m_2(1-e^{x_1-x_2})^2
+m_2m_3(1-e^{x_2-x_3})^2+m_1m_3(1-e^{x_1-x_3})^2\big]z^2\nonumber\\
&-\big[m_1m_2m_3(1-e^{x_1-x_2})^2(1-e^{x_2-x_3})^2\big] z^3\\
&B(z)=2\big[m_1e^{x_1}+m_2e^{x_2}+m_3e^{x_3}\big]z\nonumber\\
&-2\big[m_1m_2(1-e^{x_1-x_2})^2e^{x_2}
+m_2m_3(1-e^{x_2-x_3})^2e^{x_3}+m_1m_3(1-e^{x_1-x_3})^2e^{x_3}\big]z^2\nonumber\\
&+2\big[m_1m_2m_3(1-e^{x_1-x_2})^2(1-e^{x_2-x_3})^2 e^{x_3}\big]z^3,\\
&C(z)=-\big[m_1e^{2x_1}+m_2e^{2x_2}+m_3e^{2x_3}\big]z\nonumber\\
&\big[m_1m_2(1-e^{x_1-x_2})^2e^{2x_2}
+m_2m_3(1-e^{x_2-x_3})^2e^{2x_3}+m_1m_3(1-e^{x_1-x_3})^2e^{2x_3}\big]z^2\nonumber\\
&-\big[m_1m_2m_3(1-e^{x_1-x_2})^2(1-e^{x_2-x_3})^2 e^{2x_3}\big]z^3,
\end{align}
\end{subequations}

The $t$ evolution of $A,B,C$  can easily be inferred from the second equation of the
Lax pair \eqref{eq:Lax}.  One obtains:
\begin{equation}\label{eq:t-spectral}
\dot A=0,  \quad   \dot B=\frac{B}{z}-2A M_+,  \quad \dot C=-BM_+,
\text{ where }  M_+=\sum_{i=1}^n m_i e^{x_i}.   \end{equation}
We therefore see that the asymptotic conditions \eqref{eq:asymptotics}
can be implemented by requiring $A(z)=0$,
and that this condition is preserved under the time flow \eqref{eq:t-spectral},
implying that the peakon equations with $b=3$ describe an isospectral
deformation of the boundary value problem \eqref{eq:asymptotics}.
This boundary value problem can be best studied with the help of two rational functions
\begin{definition}{(Weyl functions)}
 $\omega(z)=-\frac{B(z)}{2zA(z)}, \; \zeta(z)=\frac{C(z)-B(z)}{2zA(z)}$.
\end{definition}
In this paper we will only use $\omega(z)$.  For the case of pure peakons, $m_i>0$,
it was proved in \cite{ls-cubicstring} that $\omega(z)$ is a Stieltjes transform of a measure, which subsequently played a major role in the solution of the inverse problem.  If, however, $m$ is a signed measure
then $\omega(z)$  has a more complicated structure because the spectrum
is, in general, not simple or even real.  Yet, nontrivial information about the dynamics of peakons can be extracted from
$\omega(z)$
without knowing its precise pole structure.
To this end let us establish a simple lemma which follows trivially from
equations \eqref{eq:t-spectral} and the definition of $\omega(z)$.
\begin{lemma}
\begin{equation}\label{eq:Wt}
\dot\omega(z)=\frac{\omega(z)}{z}+\frac{M_+}{z}.
\end{equation}
\end{lemma}
From explicit formulas $A(0)=1,B(0)=0$, thus implying that $0$ is a removable singular point of $\omega(z)$. Moreover, knowing the evolution of $\omega(z)$ we can readily establish the time evolution of the data involved in its partial fraction decomposition.
\begin{theorem}\label{thm:genresevolution}
Suppose the partial fraction decomposition of $\omega(z)$ is given:
\begin{equation*}
\omega(z)=\sum_j \sum_{k=1}^{d_j} \frac{b_j^{(k)}(t)}{(z-\lambda_j)^k},
\end{equation*}
where $d_j$ is the algebraic degeneracy of the eigenvalue $\lambda_j$.  Then
\begin{equation}\label{eq:b-polynomial}
b^{(k)}_j(t)=p^{(k)}_j(t)e^{\frac{t}{\lambda_j}},
\end{equation}
where $p^{(k)}_j(t)$ is a polynomial in $t$ of degree $d_j-k$ or lower, and
\begin{equation}\label{eq:Sumofbj}
\sum_j \dot b_j^{(1)}(t)=M_+.
\end{equation}
\end{theorem}
\begin{proof}
Combining the partial fraction decomposition with \eqref{eq:Wt} one gets
\begin{equation} \label{eq:dotWparts}
\frac{\omega(z)}{z}=-\frac{M_+}{z}+\sum_j \sum_{k=1}^{d_j} \frac{\dot b_j^{(k)}(t)}{(z-\lambda_j)^k}.  \end{equation}
By Residue Theorem, we have
$
0=\mathrm{Res}\left(\frac{\omega(z)}{z},\infty\right)+\mathrm{Res}\left(\frac{\omega(z)}{z},0\right)+\sum_j\mathrm{Res} \left(\frac{\omega(z)}{z},\lambda_j\right) $
where
\begin{equation*}
\mathrm{Res}\left(\frac{\omega(z)}{z},\infty\right)=0,\;\mathrm{Res}\left(\frac{\omega(z)}{z},0\right)=-M_+,\;
\mathrm{Res}\left(\frac{\omega(z)}{z},\lambda_j\right)=\dot b_j^{(1)}(t), \end{equation*} which proves (\ref{eq:Sumofbj}).

By the formulas for the coefficients in the Laurant series of equation
\eqref{eq:dotWparts} we obtain
\begin{equation*} \dot b_j^{(k)}(t)=\left.\sum_{s=k}^{d_j}\frac1{(s-k)!}\frac{\mathrm{d}^{s-k}}{\mathrm{d}z^{s-k}}\left(\frac{b_j^{(s)}(t)}{z}\right)\right|_{z=\lambda_j}
=\sum_{s=k}^{d_j}\frac{(-1)^{s-k}}{\lambda_j^{s-k+1}}b_j^{(s)}(t),  \quad 1\leq k\leq d_j.\end{equation*} In particular, we have
$\dot b_j^{(d_j)}(t)=\frac{b_j^{(d_j)}(t)}{\lambda_j}$, hence $b_j^{(d_j)}(t)=b_j^{(d_j)}(0)e^{\frac t{\lambda_j}}$. Proceeding by induction we obtain
\begin{equation*}
\dot b_j^{(k)}(t)=\sum_{s=k}^{d_j}\frac{(-1)^{s-k}}{\lambda_j^{s-k+1}}b_j^{(s)}
=\frac{b_j^{(k)}(t)}{\lambda_j}+e^{\frac t{\lambda_j}}\sum_{s=k+1}^{d_j}\frac{(-1)^{s-k}}{\lambda_j^{s-k+1}}p_j^{(s)}(t)
\stackrel{\mathrm{def}}{=}\frac{b_j^{(k)}(t)}{\lambda_j}+e^{\frac t{\lambda_j}}\tilde p^{(k+1)}(t),
\end{equation*} therefore
$b_j^{(k)}(t)=e^{\frac t{\lambda_j}}\left(b_j^{(k)}(0)+\int_0^t\tilde p^{(k+1)}(\tau)\mathrm{d}\tau\right)\stackrel{\mathrm{def}}{=}e^{\frac{t}{\lambda_j}}p^{(k)}_j(t)$. Finally, since $\tilde p^{(k+1)}(t)$ is a polynomial of degree $d_j-k-1$ or lower,
$p^{(k)}_j(t)$ is a polynomial of degree $d_j-k$ or lower, which leads to \eqref{eq:b-polynomial}.
\end{proof}

\begin{lemma} \label{cor:xn}
Let $x_n$ be the position of the $n$-th mass.
Then
\begin{equation}\label{eq:xn}
e^{x_n}=\sum_j b_j^{(1)}.
\end{equation}
\end{lemma}
\begin{proof}
By Residue Theorem
$0=\mathrm{Res}\left(\omega(z),\infty\right)+\sum_j\mathrm{Res} \left(\omega(z),\lambda_j\right)$.
Thus \\$\sum_j b_j^{(1)}=-\mathrm{Res}\left(\omega(z),\infty\right)$.
With the help of explicit formulas \eqref{eq:kABC} and the definition of $\omega(z)$
we obtain
\[\mathrm{Res}\left(\omega(z),\infty\right)=-\lim_{z\to\infty}z\omega(z)=\lim_{z\to\infty}\frac{B(z)}{2A(z)}=-e^{x_n},\] which proves the conclusion.
\end{proof}

\begin{corollary}\label{cor:noescaperight}
The $n$th mass cannot escape to $+\infty$ in finite  real time.
\end{corollary}
\begin{proof}
Indeed, from Theorem \ref{thm:genresevolution} and the
lemma above we see that $e^{x_n}$ has at most an exponential
growth, hence it is bounded for finite real time.
\end{proof}
To deal with the behaviour of $x_1$ we will use a slightly modified spectral
problem which, in principle, amounts to ``sweeping" the masses
in the opposite direction.  To this end we consider the adjoint Lax pair:
\begin{equation}\label{eq:adLax}
(\partial_x-\partial_{xxx})\tilde\Psi=-zm\tilde \Psi, \qquad
\tilde\Psi_t=[-z^{-1}(1-\partial^2_x)+u_x-u\partial_x]\tilde \Psi.
\end{equation}
\begin{remark} The only difference between equations \eqref{eq:Lax} and
\eqref{eq:adLax} is the sign of $z$ which has no effect on the compatibility
conditions; hence the adjoint Lax pair gives the same compatibility condition --- the
DP equation.
\end{remark}
We choose a different set of asymptotic conditions, namely
\begin{equation}\label{eq:adasymptotics}
\tilde\Psi\sim e^{-x},  \text{ as } x\rightarrow +\infty, \qquad  \tilde\Psi \text{ is bounded as   } x\rightarrow -\infty.
\end{equation}
For $x<x_1$
\begin{equation}\label{eq:tildeABC}
\tilde \Psi(x)=\tilde A(z)e^{-x}+\tilde B(z)+\tilde C(z)e^{x}.
\end{equation}
Hence the adjoint spectral problem is given by $\tilde A(z)=0$.
Likewise, one readily checks that the time flow given by the second equation in
\eqref{eq:adLax} yields:
\begin{equation}\label{eq:t-adspectral}
\dot {\tilde A}=0, \quad \dot {\tilde B}=-\frac{\tilde B}{z}+2\tilde AM_-, \quad \dot {\tilde C}=\tilde BM_-,\text{ where } M_-=\sum_i m_i e^{-x_i}.
\end{equation}

We conclude that the adjoint boundary value problem \eqref{eq:adasymptotics}
is also isospectral under the DP flow.  In fact, the spectral problems
\eqref{eq:asymptotics} and \eqref{eq:adasymptotics} have identical spectra.
To demonstrate that we establish first an elementary lemma.
\begin{lemma}\label{lem:tildePsi} If $\Psi(x)$ is the solution to the x-equation in the boundary value problem \eqref{eq:Lax} with $m(x)=\sum_{i=1}^nm_i\delta_{x_i}$, then $\Psi(-x)$ is the solution to the x-equation in \eqref{eq:adLax} with $\tilde m(x)=\sum_{i=1}^nm_i\delta_{-x_i}$ and
boundary conditions \eqref{eq:adasymptotics}.
\end{lemma}

\begin{proof} Since $\Psi(x)$ is the solution to (\ref{eq:Lax}), the boundary conditions \[\Psi(-x)\rightarrow e^{-x},  \text{ as } x\rightarrow +\infty, \qquad  \Psi(-x) \text{ is bounded as   } x\rightarrow -\infty\] hold. Moreover, we have \[(\partial_x-\partial_{xxx})\Psi(-x)=-(\Psi_x(-x)-\Psi_{xxx}(-x))=-zm(-x)\Psi(-x).\] Notice that $\delta_{x_i}(-x)=\delta_{-x_i}(x)$, hence $m(-x)=\tilde m(x)$ and the conclusion holds.\end{proof}

Denote $\underline{m}=(m_1,\ldots,m_n),\underline{x}=(x_1,\ldots,x_n)$ for short, and set $\underline{m}^\tau,\underline{x}^\tau$ to be the vector with the reversed order of its
entries, that is $\underline{m}^\tau=(m_n, \ldots, m_1)$ etc.  Employing the same convention as in equation \eqref{eq:AkBkCk}, but this time for $\tilde \Psi$, we obtain the following analogue of
equation \eqref{eq:kABC}.

\begin{theorem} \label{thm:symmtildeABC}
Let $\tilde A_k(z)=\tilde A_k(z;\underline{m}, \underline{x}), \, 1\leq k\leq n$.  Then
\[\left(\begin{array}{c}\tilde A_k(z)\\ \tilde B_k(z)\\ \tilde C_k(z) \end{array}\right)=\left(\begin{array}{c}1\\ 0\\ 0 \end{array}\right)+\sum_{p=1}^k\left[\sum_{I\in {{[1,k]}\choose p}}\left(\prod_{i\in I}m_i\right)\left(\prod_{j=1}^{p-1}(1-e^{x_{i_j}-x_{i_{j+1}}})^2\right)\left(\begin{array}{c}1\\ -2e^{-x_{i_1}}\\e^{-2x_{i_1}}  \end{array}\right) \right](-z)^p,\]  where $\binom{[1,n]}{p}$ is the set of all $p$-element subsets
  $I=\left\{ i_1 < \dots < i_p \right\}$ of $\left\{ 1,\dots,n \right\}$.   In particular, when
  $k=n$, $A(z)=\tilde A(z)$.
\end{theorem}

\begin{proof} By lemma \ref{lem:tildePsi} $\tilde\Psi(x)$ in the asymptotic region
$x\to -\infty$ can be expressed as \[\begin{aligned}&\tilde\Psi(x;\underline{m},\underline{x})=\tilde A(z;\underline{m},\underline{x})e^{-x}+\tilde B(z;\underline{m},\underline{x})+\tilde C(z;\underline{m},\underline{x})e^{x}\\=&\Psi(-x;\underline{m}^\tau,-\underline{x}^\tau)= A(z;\underline{m}^\tau,-\underline{x}^\tau)e^{-x}+B(z;\underline{m}^\tau,-\underline{x}^\tau)+ C(z;\underline{m}^\tau,-\underline{x}^\tau)e^{x} \end{aligned},\] which leads to $\tilde A(z;\underline{m},\underline{x})=A(z;\underline{m}^\tau,-\underline{x}^\tau),\tilde B(z;\underline{m},\underline{x})=B(z;\underline{m}^\tau,-\underline{x}^\tau)$, and $\tilde C(z;\underline{m},\underline{x})=C(z;\underline{m}^\tau,-\underline{x}^\tau).$ The conclusion then directly follows from the formulas \eqref{eq:kABC}  along with
an elementary observation that the permutation $\tau$ is a bijection on
the ordered $p$-tuples, which for any fixed $p$-tuple maps the last element $e^{x_{i_p}}$ in the
original sum into the first element $e^{-x_{n+1-i_p}}$ of the new $p$-tuple.
After a simple change of index the main claim is proven.  As to $A_n(z)$, which corresponds to the
first line in the formula, we observe that this polynomial is invariant under the transformation $\underline{m}\mapsto\underline{m}^\tau,\underline{x}\mapsto-\underline{x}^\tau$.  \end{proof}

We can thus define the adjoint Weyl function
$\widetilde \omega(z)=-\frac{\tilde B(z)}{2z\tilde A(z)}$,
and use equations \eqref{eq:t-adspectral} to determine the time flow of $\widetilde\omega$.
An easy computation gives:
\begin{lemma}
\begin{equation}
\dot{\widetilde\omega}(z)=-\frac{\widetilde \omega(z)}{z}-\frac{M_-}{z}.
\end{equation}
\end{lemma}
Consequently, we obtain an analogue of Theorem \ref{thm:genresevolution}.
\begin{theorem}\label{thm:tildegenresevolution}
Suppose the partial fraction decomposition of $\widetilde \omega(z)$ is given:
\begin{equation*}
\widetilde \omega(z)=\sum_j \sum_{k=1}^{d_j} \frac{\tilde b_j^{(k)}}{(z-\lambda_j)^k}
\end{equation*}
where $d_j$ is the algebraic degeneracy of the eigenvalue $\lambda_j$.
Then
\begin{equation}
\tilde b^{(k)}_j=\tilde p^{(k)}_j(t)e^{-\frac{t}{\lambda_j}}
\end{equation}
where $\tilde p^{(k)}_j(t)$ is a polynomial in $t$ of degree $d_j-k$ or lower, and
\begin{equation}
\sum_j \dot {\tilde b}_j^{(1)}=-M_-.
\end{equation}
\end{theorem}

With the help of Theorem \ref{thm:symmtildeABC} it is now
straightforward to establish a counterpart of Lemma \ref{cor:xn}.

\begin{lemma}\label{cor:x1}
Let $x_1$ be the position of the first mass.
Then
\begin{equation}\label{eq:x1}
e^{-x_1}=\sum_j \tilde b_j^{(1)}(t).
\end{equation}
\end{lemma}
Likewise, an analogue of Corollary \ref{cor:noescaperight} is immediate.

\begin{corollary}\label{cor:noescapeleft}
The first mass cannot escape to $-\infty$ in finite  real time.
\end{corollary}
\begin{example} Case $n=3$.
In this case we are only dealing with simple and quadratic roots, since the triple roots cannot occur as will be proved in Section 4. The formulas for $e^{x_3(t)}$ and $e^{-x_1(t)}$ read:
\begin{equation*}\begin{aligned}
&e^{x_3(t)}=\left\{\begin{aligned}&b_1^{(1)}(0)e^{\frac t{\lambda_1}}+b_2^{(1)}(0)e^{\frac t{\lambda_2}}+b_3^{(1)}(0)e^{\frac t{\lambda_3}},&\text{simple roots,} \\&b_1^{(1)}(0)e^{\frac t{\lambda_1}}+(b_2^{(1)}(0)-\frac{b_2^{(2)}(0)t}{\lambda_2^2})e^{\frac t{\lambda_2}},&\text{quadratic roots.}  \end{aligned}\right. \\
&e^{-x_1(t)}=\left\{\begin{aligned}&\tilde b_1^{(1)}(0)e^{-\frac t{\lambda_1}}+\tilde b_2^{(1)}(0)e^{-\frac t{\lambda_2}}+\tilde b_3^{(1)}(0)e^{-\frac t{\lambda_3}},&\text{simple roots,} \\&\tilde b_1^{(1)}(0)e^{-\frac t{\lambda_1}}+(\tilde b_2^{(1)}(0)+\frac{\tilde b_2^{(2)}(0)t}{\lambda_2^2})e^{-\frac t{\lambda_2}},&\text{quadratic roots.}  \end{aligned}\right. \end{aligned}
\end{equation*}

\end{example}

The spectral problem and its adjoint are clearly related and we turn now to establishing
a relation between them.  To this end we study the coefficients occurring in the eigenfunctions of the spectral problem \eqref{eq:ABC} and \eqref{eq:tildeABC}.
\begin{lemma}\label{lem:FundIdentity}
\begin{align*}
&2A(z)C(-z)+2A(-z)C(z)-B(z)B(-z)=0, \\
& 2\tilde A(z)\tilde C(-z)+2\tilde A(-z)\tilde C(z)-\tilde B(z)\tilde B(-z)=0.
\end{align*}
\end{lemma}
\begin{proof}
It suffices to write $(D-D^3)\psi(x;z) =zm \psi(x;z), (D-D^3)\psi(x;\lambda) =\lambda m \psi(x;\lambda)$ and obtain from it the identity:
\begin{equation*}
D\big( \psi(x;z)\psi(x;\lambda)-B(\psi(x;z),\psi(x;\lambda))\big)=(\lambda+z)m\psi(x;z)\psi(x;\lambda),
\end{equation*}
where $B(f,g)=f''g-f'g'+fg''$.  Finally, if one sets $z+\lambda=0$ and evaluates
the above expression at $x\rightarrow -\infty$ and $x\rightarrow \infty$ one obtains the first
claim. The proof of the second identity is analogous.
\end{proof}
We will briefly study the symmetry responsible for the connection between
the boundary value problem \eqref{eq:asymptotics} and its adjoint \eqref{eq:adasymptotics}.  To this end we recall the {\sl transition} matrix $S(z)$
introduced in \cite{ls-cubicstring}
\begin{equation*}
S(z)=S_n(z)S_{n-1}(z)\dotsb S_1(z),  \text{ where }
\begin{bmatrix}A_k\\B_k\\C_k\end{bmatrix}=S_k(z)\begin{bmatrix}A_{k-1}\\B_{k-1}\\C_{k-1}\end{bmatrix},
\end{equation*}
where $A_0=1, B_0=C_0=0$.  An explicit form of $S_k(z)$ is easy to compute:
\begin{equation}\label{eq:Sk}
S_k(z)=I-zm_k\begin{bmatrix}e^{-x_k}\\-2\\e^{x_k}\end{bmatrix}
\begin{bmatrix}e^{x_k}&1&e^{-x_k}\end{bmatrix}.
\end{equation}
Define now
\begin{definition}
$J=\begin{bmatrix} 0&0&1\\0&-2&0\\1&0&0 \end{bmatrix}$.
\end{definition}
We can define the loop group of continuous maps $G: \R \rightarrow \text{SL}(3,\R)$; clearly $S(z) \in G$.  Moreover, if we introduce involution:
$\tau: G\rightarrow G, g(z)\rightarrow J(g^{-1}(-z))^T J^{-1}$, then
$S_k(z)\in G_{\tau}=\{g=\tau(g)\}$, a subgroup fixed by $\tau$.  Hence
\begin{lemma}
$S(z) \in G_{\tau}$.
\end{lemma}
Let us denote the canonical basis $e_1=\begin{bmatrix}1\\0\\0 \end{bmatrix},
e_2=\begin{bmatrix}0\\1\\0 \end{bmatrix}, e_3=\begin{bmatrix}0\\0\\1 \end{bmatrix}
$ by the shorthand notation $|1\rangle|, |2\rangle, |3\rangle$.  Then equations
\eqref{eq:ABC}, \eqref{eq:tildeABC} can be written
\begin{equation}
\begin{bmatrix}A(z)\\B(z)\\C(z)\end{bmatrix}=S(z)|1\rangle
\qquad \begin{bmatrix}\tilde C(z)\\\tilde B(z)\\\tilde A(z)\end{bmatrix}=\widetilde S(z)|3\rangle,
\end{equation}
where
\begin{equation}\label{eq:tildeS}
\widetilde S(z)=S_1^{-1}(-z)\dotsb S_n^{-1}(-z)=S(-z)^{-1}.
\end{equation}
Moreover, since $S_k(z)\in G_{\tau}$,
\begin{equation}\label{eq:tildeSg}
\widetilde S(z)=JS^T(z)J^{-1}.
\end{equation}
This is a fundamental relation which allows one to relate the spectral data
for the boundary problem and its adjoint.
\begin{theorem}\label{thm:symmetry}
\mbox{}
\begin{enumerate}
\item $\tilde A(z)=S(z)_{11}, \quad \tilde B(z)=-2S(z)_{12},\quad  \tilde C(z)=S(z)_{13}$.
\item Suppose $\lambda_i$ is a root of $A(z)=0$ then
\begin{equation}\label{eq:cbtilde}
B(-\lambda_i)=C(\lambda_i)\tilde B(\lambda_i).
\end{equation}
\item Suppose $B(-\lambda_i)\neq 0$ then
\begin{equation}\label{eq:bbtilde}
2A(-\lambda_i)=B(\lambda_i)\tilde B(\lambda_i).
\end{equation}
\end{enumerate}
\end{theorem}
\begin{proof}
By definition $\tilde A(z)=\langle 3| JS^T(z) J^{-1}| 3 \rangle =\langle 1|S^T(z)|1 \rangle=S(z)_{11}$.  Likewise, $\tilde B(z)=\langle 2|JS^T(z)J^{-1}|3\rangle=-2\langle 2|S^T(z)|1\rangle=-2 S(z)_{12}$ and $\tilde C(z)=\langle 1|JS^T(z)J^{-1}|3\rangle=S(z)_{13} $.  The second item is proved by making
use of the involution $\tau$.  On one hand $B(z)=\langle 2|S(z)|1\rangle$, on the other, since $S(z)\in G_{\tau}$, $B(z)=-2\langle 3|S^{-1}(z)|2\rangle=2\begin{vmatrix}S(-z)_{11}&S(-z)_{12}\\S(-z)_{31}&S(-z)_{32} \end{vmatrix}$.  Finally, since $S(z)_{11}=A(z)$, evaluating the determinant at
the (minus) root $\lambda _i$ of $A(z)$ we obtain $B(-\lambda_i)=-2 S(\lambda_i)_{12}S(\lambda_i)_{31}=\tilde B(\lambda_i)C(\lambda_i)$, in view of the
statement from item $(1)$.  Finally, to prove item $(3)$, we set $z=\lambda_i$ in the statement of Lemma \ref{lem:FundIdentity} to get $2A(-\lambda_i)C(\lambda_i)=B(\lambda_i) B(-\lambda_i)$.  Upon
multiplying equation \eqref{eq:cbtilde} by $2A(-\lambda_i)$  and eliminating
the term involving $2A(-\lambda_i)C(\lambda_i)$
we obtain
 $2A(-\lambda_i)B(-\lambda_i)=B(-\lambda_i)B(\lambda_i)\tilde B_i(\lambda_i)$,
 resulting in equation \eqref{eq:bbtilde}.
\end{proof}
Consider now the Weyl function $\ds{\omega}(z)$ and its adjoint $\ds{\widetilde\omega(z)}$
in the case of simple spectrum, i.e. \\
$\ds{\omega(z)=-\frac{B(z)}{2zA(z)}=\sum_{i=1}^n \frac{b_i}{z-\lambda_i}}$ and
$\ds{\widetilde\omega(z)=-\frac{\tilde B(z)}{2z\tilde A(z)}=\sum_{i=1}^n\frac{\tilde b_i}{z-\lambda_i}}$.
\begin{theorem}\label{thm:b-tildeb} Suppose the spectral problem $A(z)=0$ has only simple roots $\lambda_i$
and there are no anti-resonances $(\lambda_i+\lambda_j\neq 0)$.  Then
\mbox{}
\begin{equation}\label{eq:b-tildeb}
b_i\tilde b_i=\prod_{j\neq i} \frac{1+\frac{\lambda_i}{\lambda_j}}{(1-\frac{\lambda_i}{\lambda_j})^2}.
\end{equation}
\end{theorem}
\begin{proof}
Under the assumption of simple spectrum:
\begin{equation*}
b_i\tilde b_i=\frac{B(\lambda_i)\tilde B(\lambda_i)}{4 \lambda_i^2  (A'(\lambda_i))^2}
\end{equation*}
which simplifies, after using equation \eqref{eq:bbtilde}, to
\begin{equation*}
b_i\tilde b_i=\frac{A(-\lambda_i)}{2 \lambda_i^2  (A'(\lambda_i))^2}=\frac{\prod_{j=1}^n (1+\frac{\lambda_i}{\lambda_j})}{2\prod_{j\neq i}(1-\frac{\lambda_i}{\lambda_j})^2},
\end{equation*}
which implies the claim if one observes that the term with $j=i$ appearing in the numerator contributes the factor of $2$ canceling the one from the numerator.
\end{proof}

\begin{remark}
This beautiful identity generalizes the one known from the
ordinary string problem \cite{bss-string} which in our notation reads:
\begin{equation*}
b_i\tilde b_i=\prod_{j\neq i}\Big(1-\frac{\lambda_i}{\lambda_{j}}\Big)^{-2}.
\end{equation*}
\end{remark}
\begin{remark}
The presence of anti-resonances $(\lambda_i+\lambda_j=0)$ is characteristic  of the DP equation as can be seen, for example, from explicit solutions.
\end{remark}
\section{Three multipeakons}\label{sec:3peakons}
In this section we apply the methods developed in Section \ref{sec:Fundamentals} to
study three multipeakons, with emphasis on the behaviour of solutions at the time of blow-up.
As before we use the multipeakon ansatz \eqref{eq:peakon-ansatz}
\begin{equation}\label{eq:3ansatz}
u(x,t)=\sum_{i=1}^3m_i(t)\, e^{-|x-x_i(t)|}
\end{equation}
where $x_1(0)<x_2(0)<x_3(0)$,  and we no longer assume that $m_i(t)$ are all
positive.  In spite of that we will refer to $m_j$s as {\sl masses} to emphasize
their roles in the spectral problem.  We will need a bit of terminology regarding
the phenomenon of breaking. Since we will be analyzing a system of ODEs
obtained from a restriction of equation \eqref{eq:b-ode-short} we will say that
at some time $t_0$ a {\sl collision } occured if for some $i\ne j$, $x_i(t_0)=x_j(t_0)$.
In the case of the CH equation the presence of a collision is
tantamount to a wave breaking (\cite{bss-moment}) but the solution can be
continued with the preservation of the Sobolev $H^1(\R)$ norm beyond the collision time.
This is not the case for the DP equation as was anticipated by Lundmark in \cite{lundmark-shockpeakons} for the case of the peakon-antipeakon pair.  We confirm his assertion that
the shockpeakons are created by proving that  $m=u-u_{xx}$
tends to the shockpeakon data in the distribution topology at the collision time (see Theorem \ref{thm:shockcreation}).

We start by setting $b=3$ and $n=3$ in the multipeakon equation
\eqref{eq:b-ode-short}, which leads to the following ODEs in the sector $X=\{{\bf x} \in \R^3\ |x_1<x_2<x_3\}$:
\begin{subequations}\label{eq:3peakons}
\begin{align}
&\dot{x}_1=m_1+m_2e^{x_1-x_2}+m_3e^{x_1-x_3}, \\ &\dot{x}_2=m_1e^{x_1-x_2}+m_2+m_3e^{x_2-x_3}, \\ &\dot{x}_3=m_1e^{x_1-x_3}+m_2e^{x_2-x_3}+m_3, \\ &\dot{m}_1=2m_1(-m_2e^{x_1-x_2}-m_3e^{x_1-x_3}), \\ &\dot{m}_2=2m_2(m_1e^{x_1-x_2}-m_3e^{x_2-x_3}), \\ &\dot{m}_3=2m_3(m_1e^{x_1-x_3}+m_2e^{x_2-x_3}).  \end{align} \end{subequations}
This system of equations has the following obvious symmetry.
\begin{lemma}\label{lem:symmetry}
Suppose $\{x_1(t),x_2(t),x_3(t),m_1(t),m_2(t),m_3(t)\}$ is a solution
of equations \eqref{eq:3peakons} at time $t$ with the initial condition $\{x_1(0),x_2(0),x_3(0),m_1(0),m_2(0),m_3(0)\}$.
Then $\{x_1(t),x_2(t),x_3(t),-m_1(t),-m_2(t),-m_3(t)\}$ is the solution
at time $-t$ with the initial condition $\{x_1(0),x_2(0),x_3(0),-m_1(0),-m_2(0),-m_3(0)\}$.
\end{lemma}
\begin{remark} In short, the lemma above means that $t\mapsto -t, m_i\mapsto -m_i$
is a symmetry of equations \eqref{eq:3peakons} which preserves the sector $X$.
\end{remark}
Another very useful property of equations \eqref{eq:3peakons} is the
existence of three constants of motion.  Indeed, we recall that the polynomial $A(z)$
introduced in \eqref{eq:3ABC} is time invariant.
Writing
\begin{equation}\label{eq:charpoly}
A(z)=1-M_1z+M_2z^2 -M_3 z^3
\end{equation}
we obtain the following lemma.
\begin{lemma}\label{lem:constants}
$M_1, M_2, M_3$, given by:
\[\begin{aligned}&M_1=m_1+m_2+m_3,\\ &M_2=m_1m_2(1-e^{x_1-x_2})^2+m_2m_3(1-e^{x_2-x_3})^2+m_3m_1(1-e^{x_1-x_3})^2,\\ &M_3=m_1m_2m_3(1-e^{x_1-x_2})^2(1-e^{x_2-x_3})^2,  \end{aligned}\]
are constants of motion of the system of equations \eqref{eq:3peakons}.
\end{lemma}

These constants will be one of our basic tools for studying collisions.
We observe that, geometrically speaking, a collision occurs if the solution approaches the boundary
of $X$ in finite time.  This is the only singular behaviour of the system
\eqref{eq:3peakons} happening in the coordinate space since Corollaries \ref{cor:noescaperight} and \ref{cor:noescapeleft} exclude an escape scenario in finite real time. However, the shape of the constants of motion shows that at a collision
at least two masses diverge, which will be proved in Corollary \ref{cor:notriplemassdiv}.


We begin now our study of the dynamics of three multipeakons in a vicinity of the collision by
first concentrating on the particles with labels $1$ and $3$.
Lemma \ref{cor:xn}, in particular equation \eqref{eq:xn},
gives us explicit form of $x_3$:
\begin{equation}\label{eq:x3}
e^{x_3(t)}=\sum_j b_j^{(1)}(t).
\end{equation}
Thus we obtain:
\begin{lemma}\label{lem:analx3}
Let $T_1$, $T_2$, be the largest negative, respectively the smallest
positive root of $\sum_j b_j^{(1)}(t)=0$ (if $T_1$ or $T_2$ does not exist we set
$T_1=-\infty, T_2=+\infty$ respectively).  Then
$e^{x_3(t)}$ is real analytic for $T_1<t<T_2$. Moreover,
if either $T_1$ or $T_2$ are finite then there must be a collision at some prior time $T_1<t_c<T_2$.
\end{lemma}
\begin{remark} For positive $t$ ``prior" has the usual meaning (positive orientation).
 For negative $t$
the orientation is from $0$ to $-\infty$.
\end{remark}
\begin{proof}
The last statement follows from Corollary \ref{cor:noescapeleft}, since finite
$T_1$ or $T_2$ means that $x_3$ escaped to $-\infty$ which cannot happen in
finite time unless there is a collision at an earlier time.
\end{proof}
\begin{remark}
Since $b_j^{(1)}(t)$ is an exponential function of $t$, the right hand side of \eqref{eq:x3} is well define for any real $t$. However the left hand side of \eqref{eq:x3} only make sense when $t$ lies in the existence interval of the ODE system \eqref{eq:3peakons}.
\end{remark}
Furthermore, combining equations
for $\dot x_3$ and $\dot m_3$ (see equations \eqref{eq:3peakons}) yields:
\begin{equation}\label{eq:m3x3}
\dot m_3=2m_3(\dot x_3-m_3).
\end{equation}
We remark that this is a Bernoulli type equation which
can be easily solved once $x_3(t)$ is known.
\begin{lemma}\label{lem:m3}
Suppose $x_3(t)$ is known.  Then
\begin{equation}\label{eq:m3}
\frac{1}{m_3(t)}=e^{-2x_3(t)}\Big[\frac{e^{2x_3(0)}}{m_3(0)}+2\int_0^t e^{2x_3(\tau)}\mathrm{d}\tau\Big]
\end{equation}
\end{lemma}
An analogous argument works for $x_1$. Indeed,
by equation \eqref{eq:x1}, we
have
\begin{equation*}
e^{-x_1(t)}=\sum_j \tilde b_j^{(1)}.
\end{equation*}
This prompts an analogous statement to Lemma \ref{lem:analx3}
\begin{lemma}\label{lem:analx1}
Let $\tilde T_1$, $\tilde T_2$, be the largest negative, respectively the smallest
positive, root of $\sum_j \tilde b^{(1)}(t)=0$.    Then
$e^{-x_1(t)}$ is real analytic for $\tilde T_1<t<\tilde T_2$. Moreover,
if either $\tilde T_1$ or $\tilde T_2$ are finite then there must be a collision at some prior time $\tilde T_1<t_c<\tilde T_2$.
\end{lemma}
We see that we can now narrow down the time of a collision.
Let us denote by $A=(T_1,T_2)\cap (\tilde T_1,\tilde T_2).$
We summarize analytic properties of $e^{x_3}$ and $e^{-x_1}$.
\begin{lemma} \label{lem:analx1x3}
The functions $e^{x_3(t)}$ and $e^{-x_1(t) }$ are real analytic on $A$.
A collision can only occur  at a time $t_c$ if $t_c \in A$ .  In particular, both functions are analytic
at the time of collision.
\end{lemma}

Once again, if we know $x_1$ then we can determine $m_1$.

Indeed, equations \eqref{eq:3peakons} imply another Bernoulli equation:
\begin{equation}\label{eq:m1x1}
\dot m_1=-2m_1(\dot x_1-m_1),
\end{equation}
whose solution reads
\begin{lemma}\label{lem:m1}
\begin{equation}\label{eq:m1}
\frac{1}{m_1(t)}=e^{2x_1(t)}\Big[\frac{e^{-2x_1(0)}}{m_1(0)}-2\int_0^t e^{-2x_1(\tau)}\mathrm{d}\tau\Big].
\end{equation}
\end{lemma}

We can now summarize analytic properties of $m_1,m_2$ and $m_3$.

\begin{lemma}\label{lem:analm1m3}
\mbox{}
\begin{enumerate}
\item[(1)] $\frac{1}{m_1(t)}$  and $\frac{1}{m_3(t)}$ are real analytic on $A$.
\item[(2)] $m_2(t)$ and $e^{x_2(t)}$ are real meromorphic functions on $A$.
\item[(3)] A collision occurs iff
there exists $t_c \in A$ such that either $\frac{1}{m_1(t_c)}=0$
or $\frac{1}{m_3(t_c)}=0$.
\item[(4)] Suppose $\frac{1}{m_1(t_c)}=0$.  Then in a neighborhood of $t_c$
\begin{equation*}
\frac{1}{m_1(t)}=-2 (t-t_c) +\Oh((t-t_c)^2).
\end{equation*}
\item[(5)] Suppose $\frac{1}{m_3(t_c)}=0$.  Then in a neighborhood of $t_c$
\begin{equation*}
\frac{1}{m_3(t)}=2 (t-t_c) +\Oh((t-t_c)^2).
\end{equation*}
\end{enumerate}
\end{lemma}
\begin{proof}
The analytic properties of $\frac1{m_1}$ and $\frac1{m_3}$ are directly derived from the analytic properties of $e^{x_1(t)}$ and $e^{x_3(t)}$ and Lemmas \ref{lem:m1}, \ref{lem:m3} respectively.

To prove (2) we note that using $M_1$ we can write $m_2=M_1-(m_1+m_3)$.
Likewise $e^{x_2}$ can be easily computed from $M_+=\sum _j m_j e^{x_j}=\frac{\mathrm{d}}{\mathrm{d}t}{e^{x_3}}$ by algebraic operations on analytic functions.

To see (3), since $M_3$ is a constant of motion, we observe that a collision occurs when at least two of the masses diverge,
hence either $m_1$ or $m_3$ have to diverge at a collision.

To prove (4) and (5) we only need to calculate the derivatives of $\frac{1}{m_1(t)}$ and $\frac{1}{m_3(t)}$ at $t_c$. Since $x_1,x_3$ are analytic at $t_c$, direct computation from (\ref{eq:m3x3}) and (\ref{eq:m1x1}) shows that
\begin{equation*} \begin{aligned} &\left.\frac{\mathrm{d}}{\mathrm{d}t}\frac{1}{m_3(t)}\right|_{t=t_c}=
\left.\left(2-\frac{2\dot x_3}{m_3}\right)\right|_{t=t_c}=2,\\ &\left.\frac{\mathrm{d}}{\mathrm{d}t}\frac{1}{m_1(t)}\right|_{t=t_c}=
\left.\left(\frac{2\dot x_1}{m_1}-2\right)\right|_{t=t_c}=-2.\end{aligned} \end{equation*} Therefore (3) and (4) hold.

\end{proof}

Now we can obtain the behaviour of masses before a collision.
As a general comment, we observe that for any
initial data in $X$ and arbitrary $m_1,m_2,m_3$ the solution is unique.
From this point onwards we assume $m_i(0)\neq 0$.
Thus $M_3$ is nonzero.

\begin{lemma}  \label{lem:mzero} None of the masses $m_i$ can become zero before a collision. \end{lemma}

\begin{proof}
In order to derive a contradiction, we suppose that one of the masses becomes zero at $t_0$.  Since all three constants of motion given by Lemma \ref{lem:constants} are symmetric with respect to permutations of masses, we can assume, without a loss of generality, that $m_1(t_0)=0$.  Since $M_3\neq 0$, $m_2m_3$ diverges at $t_0$. Then $\frac{M_2}{m_2m_3}$ converges to zero, while the corresponding right hand
side converges to $(1-e^{x_2-x_3})^2\neq 0$, thus a contradiction.
\end{proof}

\begin{corollary} The masses $m_i$ cannot change their signs before a collision.  \end{corollary}

\begin{corollary} None of the masses $m_i$ will diverge to $\pm \infty$ before a collision.
\end{corollary}
\begin{proof}
Since none of $m_i$ can become zero by Lemma \ref{lem:mzero}, $M_3$ is nonzero, $0<\frac{M_3}{m_1m_2m_3}<\infty$ and $\frac{M_3}{m_1m_2m_3}=(1-e^{x_1-x_2})^2(1-e^{x_2-x_3})^2$ is continuous before a collision, hence the claim
follows.
\end{proof}
Combined with the analytic property, there are two corollaries worth mentioning.
\begin{corollary}[Absence of triple collisions]\label{cor:notriplecol}
There are no triple collisions, that is, there is no time at which $x_1=x_2=x_3$.
\end{corollary}
\begin{proof}
Suppose $x_1(t_c)=x_2(t_c)=x_3(t_c)$.  Then the leading contribution
to $M_3$ coming from the term $(1-e^{x_1-x_2})^2(1-e^{x_2-x_3})^2$
is $\Oh((t-t_c)^4)$ which forces $m_1m_2m_3$ to behave like $\Oh(\frac{1}{(t-t_c)^4})$.
If only $m_1, m_2$ diverge then $m_1$ diverges as $\frac{1}{2(t_c-t)}$ by Lemma
\ref{lem:analm1m3} and $m_2$ would have to diverge as $\Oh(\frac{1}{(t-t_c)^3})$
violating conservation of $M_1$. Similar argument excludes divergence of $m_2,m_3$.  The last case is that $m_1,m_2, m_3$ diverge, but
then in view of Lemma \ref{lem:analm1m3} $m_2$ would have to
diverge as $\Oh(\frac{1}{(t-t_c)^2})$, again violating conservation of $M_1$.
\end{proof}

\begin{corollary}\label{cor:notriplemassdiv}
At the point of a collision masses diverge in pairs and the only admissible pairs are $\{m_1,m_2\}$ and $\{m_2,m_3\}$.
\end{corollary}
\begin{proof}
In view of the behaviour of $m_1$ and $m_3$ at the collision, $m_2$ must be
regular  to preserve $M_1$ if $m_1$  and $m_3$ diverge. Thus $m_1,m_2,m_3$
cannot all diverge.  To eliminate the $m_1, m_3$ pair we consider $0=\lim_{t\rightarrow t_c} \frac{M_2}{m_1m_3}=(1-e^{x_1(t_c)-x_3(t_c)})^2\neq 0$ by the
absence of triple collisions, hence a contradiction.
\end{proof}

\begin{theorem}[Shockpeakon creation] \label{thm:shockcreation}
If $m_j$ collides with $m_{j+1}$ at $t_c>0$, then
\[\begin{aligned}&\lim_{t\to t_c^-}(m_j(t)\delta(x-x_j(t))+m_{j+i}(t)\delta(x-x_{j+1}(t))) \\=&\left(\lim_{t\to t_c^-}(m_j+m_{j+1})\right)\delta(x-x(t_c))+\frac12\left(\lim_{t\to t_c^-}(u(x_j(t),t)-u(x_{j+1}(t),t))\right)\delta'(x-x(t_c)),\end{aligned}\]
where the limit is in the sense of $\mathscr{D}'(\mathbb{R})$.
\end{theorem}
\begin{proof}
For arbitrary $\varphi(x)\in\mathscr{D}(\mathbb{R})$,
\[\langle m_j(t)\delta(x-x_j(t))+m_{j+i}(t)\delta(x-x_{j+1}(t)),\varphi(x)\rangle=m_j(t)\varphi(x_j(t))+m_{j+i}(t)\varphi(x_{j+1}(t)).\]
Whenever $j=1$ or $2$, we can always write
\[m_j=-\frac1{2(t-t_c)}+C_0+O(t-t_c),\quad m_{j+1}=\frac1{2(t-t_c)}+\tilde C_0+O(t-t_c) \] around $t_c$. Hence,
\[\begin{aligned}&\lim_{t\to t_c^-}\langle m_j(t)\delta(x-x_j(t))+m_{j+i}(t)\delta(x-x_{j+1}(t)),\varphi(x)\rangle \\
=&(C_0+\tilde C_0)\varphi(x(t_c))-\lim_{t\to t_c}\frac{\varphi(x_j(t))-\varphi(x_{j+1}(t))}{2(t-t_c)} \\
=&\left(\lim_{t\to t_c^-}(m_j+m_{j+1})\right)\varphi(x(t_c))-\frac12\left(\lim_{t\to t_c^-}(\dot{x}_j-\dot{x}_{j+1})\right)\varphi'(x(t_c)) \\
=&\left(\lim_{t\to t_c^-}(m_j+m_{j+1})\right)\varphi(x(t_c))-\frac12\left(\lim_{t\to t_c^-}(u(x_j(t),t)-u(x_{j+1}(t),t))\right)\varphi'(x(t_c)),
\end{aligned}\] where in the last step we have used equation \eqref{eq:b-ode-short} for $b=3$.  The claim follows now easily from the definitions of distributions $\delta$ and $\delta '$.
\end{proof}
\begin{remark}
Shockpeakon creation described by Theorem \ref{thm:shockcreation} confirms
the scenario that at the collision the colliding peakon-antipeakon pair creates
the shock (the $\delta'$ contribution above) and the peakon or antipeakon contribution  (the $\delta$ contribution) thus
giving the overall collision data of two peakons/antipeakons and a shock.  This has been previously verified for the case $n=2$ in \cite{lundmark-shockpeakons}.
\end{remark}
\section{Three multipeakons; spectral properties} \label{sec:3peakspectral}
This section addresses basic questions related to the spectral characterization of the
peakon dynamics \eqref{eq:3peakons}.
\begin{lemma}\label{lem:signeigen}
Let $N^+$ denote the number of positive masses and $n^+$ be the number of
eigenvalues of the spectral problem $A(z)=1-M_1z+M_2z^2-M_3z^3=0$
which have strictly positive real parts.  Then
\begin{equation*}
N^+=n^+.
\end{equation*}
\end{lemma}

\begin{proof} The statement holds true if $N^+=3$ by results in
\cite{ls-cubicstring}; in that case the spectrum is positive and simple.
Since $m_i\mapsto-m_i,\lambda_i\mapsto-\lambda_i$ is a symmetry of the eigenvalue problem it suffices to analyze only the case with two positive masses, that is $N^+=2$.


Then \[M_3=\frac1{\lambda_1\lambda_2\lambda_3}<0. \] To prove the claim we
have to exclude that three eigenvalues have strictly negative real parts (recalling that complex roots must occur
in conjugate pairs) or that there is one negative and two purely imaginary conjugate roots.  In either case $M_1=\frac1{\lambda_1}+\frac1{\lambda_2}+\frac1{\lambda_3}<0,\;\text{and}\; M_2=\frac1{\lambda_1\lambda_2}+\frac1{\lambda_2\lambda_3}+\frac1{\lambda_3\lambda_1}>0.$ So we assume that $M_1<0$ and $M_2>0$ in order to derive a contradiction.
\vspace{0.5cm}

\noindent Case 1. If $m_1<0$, then $m_2,m_3>0$ and $0<m_2+m_3<-m_1$. Therefore \[-m_1m_3(1-e^{x_1-x_3})^2>m_2m_3(1-e^{x_2-x_3})^2, \] which implies \[M_2=m_1m_2(1-e^{x_1-x_2})^2+m_2m_3(1-e^{x_2-x_3})^2+m_1m_3(1-e^{x_1-x_3})^2<0 \] and thus leads to a contradiction.

\noindent Case 2. If $m_3<0$, it is similar to Case 1.

\noindent Case 3. If $m_2<0$, then $m_1,m_3>0$ and $0<m_1+m_3<-m_2$.

Denote $\hat m=m_1+m_3$, then  \[m_1m_2(1-e^{x_1-x_2})^2+m_2m_3(1-e^{x_2-x_3})^2<-\hat m[m_1(1-e^{x_1-x_2})^2+m_3(1-e^{x_2-x_3})^2].\]
Set $\alpha=e^{x_1-x_2},\beta=e^{x_2-x_3},m_1=\theta \hat m,m_3=(1-\theta)\hat m$, and
\[f(\theta)=\theta(1-\alpha)^2+(1-\theta)(1-\beta)^2-\theta(1-\theta)(1-\alpha\beta)^2.  \] $f(\theta)$ is a quadratic function with respect to $\theta$ with the discriminant \[\begin{aligned}&\Delta=((1-\alpha)^2-(1-\beta)^2-(1-\alpha\beta)^2)^2-4(1-\alpha\beta)^2(1-\beta)^2\\ &=-(1-\alpha)^2(1-\beta)^2(1+\alpha)(1+\beta)(3-\alpha-\beta-\alpha\beta)<0, \end{aligned}\] which leads to $f(\theta)>0$. Therefore \[\begin{aligned}&M_2<-\hat m[m_1(1-e^{x_1-x_2})^2+m_3(1-e^{x_2-x_3})^2]+m_3m_1(1-e^{x_1-x_3})^2 \\&=-\hat m^2(\theta(1-\alpha)^2+(1-\theta)(1-\beta)^2-\theta(1-\theta)(1-\alpha\beta)^2)=-\hat m^2f(\theta)<0, \end{aligned}\]
hence a contradiction.
\end{proof}

Clearly, by reflection symmetry, we obtain
\begin{corollary}
Let $N^-$ denote the number of negative masses and $n^-$ be the number of
eigenvalues of the spectral problem $A(z)=1-M_1z+M_2z^2-M_3z^3=0$
which have strictly negative  real parts.  Then
\begin{equation*}
N^-=n^-.
\end{equation*}
\end{corollary}
  Another useful corollary is that there are no eigenvalues on the line $\textrm{Re} \, z=0$.
\begin{corollary} \label{cor:nopureimag}
None of the eigenvalues of the spectral problem $A(z)=1-M_1z+M_2z^2-M_3z^3=0$ is
purely imaginary.
\end{corollary}
\begin{corollary} \label{cor:tripleroots}
The spectral problem for $n=3$ can never have triple roots.
\end{corollary}
\begin{proof}
Suppose, without loss of generality, that the spectral problem has triple positive roots.  Then all the masses are positive, i.e. the peakons case. However the eigenvalues for the peakons are simple, hence a contradiction.
\end{proof}
\begin{remark} Figure \ref{fig:1} on page \pageref{fig:1} illustrates how the eigenvalues are distributed
for the mass signature $m_1(0)>0, m_2(0)<0, m_3(0)>0$, abbreviated \\($+-+$).
The graph depicts $75\times 75$ triples of
eigenvalues for different values of masses within that configuration.
The actual input data is $m_1=1.2+0.02j, \, m_2 =-5-0.01k, \,  m_3=4, \, x_1=-0.2, \, x_2=0, x_3=0.1,\, 1\leq j,k \leq 75$.
Observe that indeed the line $\textrm{Re}\, \lambda=0$ contains no eigenvalues.
\end{remark}
\begin{center}
\begin{figure}[heretp]
\includegraphics*[viewport=60 60 600 600, width=12cm,height=12cm]{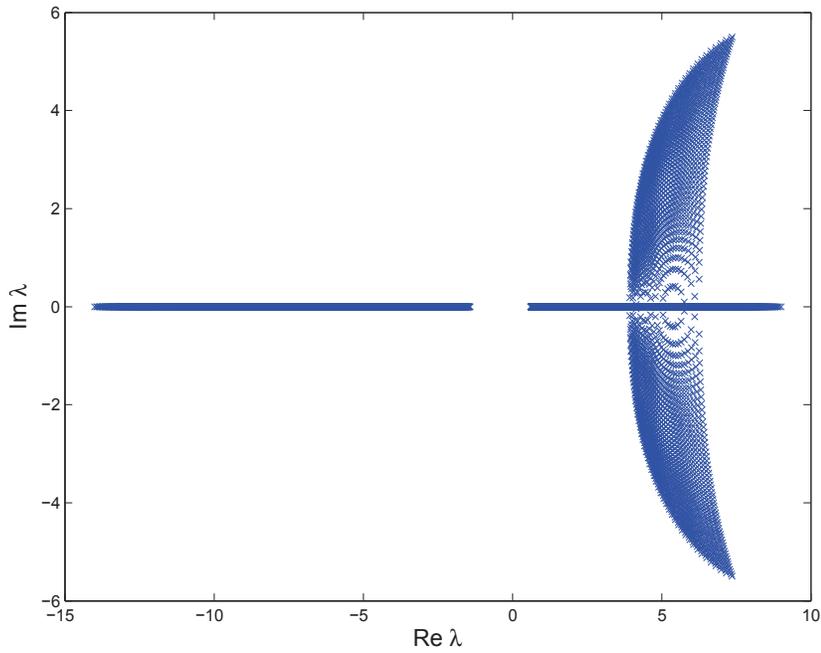}
\caption{A portrait of eigenvalue distribution for the mass signature $(+-+)$}
\label{fig:1}
\end{figure}
\end{center}
\newpage

\section{Three multipeakons; classification}\label{sec:3peakclass}

%
The goal of this section is to understand the impact of the configuration of
signs of masses on the occurrence of collisions.

One can classify the system of three multipeakons by the signs of the initial values of $m_i$s: \[\begin{aligned} &\text{(i)}\,m_1(0)>0, m_2(0)>0, m_3(0)>0;&&\text{(ii)}\,m_1(0)>0, m_2(0)>0, m_3(0)<0;\\ &\text{(iii)}\,m_1(0)>0, m_2(0)<0, m_3(0)>0;&&\text{(iv)}\,m_1(0)>0, m_2(0)<0, m_3(0)<0;\\ &\text{(v)}\,m_1(0)<0, m_2(0)>0, m_3(0)>0;&&\text{(vi)}\,m_1(0)<0, m_2(0)>0, m_3(0)<0;\\ &\text{(vii)}\,m_1(0)<0, m_2(0)<0, m_3(0)>0;&&\text{(viii)}\,m_1(0)<0, m_2(0)<0, m_3(0)<0.\\\end{aligned}\] The first and last cases are pure peakon and antipeakon, which are already well-known. In view of Lemma \ref{lem:symmetry}, the symmetry
$m_i \rightarrow -m_i, t\rightarrow -t$ reduces  the eight cases
to four cases.  To gain some clarity we will supplement a reference to
any of the cases from the list above by an ordered collection of signs, i.e.
case (i) is equivalent to $(+++)$, case (ii) to $(++-)$ etc. and we will refer to
a given mass signature as a {\sl mass signature}.

\begin{theorem}\label{thm:collisions}
If the mass signature is:
\begin{itemize}
\item[(1)] $(--+)$ or $(-++)$, then no collisions will occur for positive times,
\item[(2)] $(++-)$ or $(+--)$, then no collision will occur for negative times.
\end{itemize}
Furthermore, if the eigenvalues are not in anti-resonance, a collision will always happen at some finite time $t_c$
\begin{itemize}
\item[(1)] if the mass signature is $(--+)$ or $(-++)$, then the collision will happen at a negative time,
\item[(2)] if the mass signature is $(++-)$ or $(+--)$, then the collision will happen at a positive time.
\item[(3)] if the mass signature is $(-+-)$ or $(+-+)$, then the collision will happen at both a finite positive and a finite negative time.
\end{itemize}
\end{theorem}
\begin{proof} First, we show that in the case of item $(1)$ no collisions occur
in positive time.  Indeed, by examining the formulas \eqref{eq:m3} and
\eqref{eq:m1} we see that the respective right hand sides can not be
equal $0$ for $t\geq 0$.  The same argument works for item $(2)$ and
negative times.

Let $\lambda_1,\lambda_2,\lambda_3$ be the eigenvalues of the system which we can order as
$\mathrm{Re}\,\frac{1}{\lambda_1}\leq\mathrm{Re}\,\frac{1}{\lambda_2}\leq\mathrm{Re}\,\frac{1}{\lambda_3}$.   Since the masses have different signs, by Lemma \ref{lem:signeigen}, we have $\mathrm{Re}\,\frac{1}{\lambda_1}<0<\mathrm{Re}\,\frac{1}{\lambda_3}$.

\textit{Case 1}: The eigenvalues are simple. Since the eigenvalues are not in anti-resonance, all the residues $b_i$'s and $\tilde b_i$'s are nonzero according to \eqref{eq:b-tildeb}. Hence, according to (\ref{eq:xn}) and (\ref{eq:x1}), there exists at least one increasing and one decreasing exponential function in the expansions of $e^{x_3}$ (respectively in the expansions of $e^{-x_1}$). Moreover, since we are squaring $e^{x_3}$, $e^{-x_1}$ respectively, the
coefficient of the leading exponential will be strictly positive if the spectrum
is real, or strictly positive except for a set of measure zero if
the spectrum is degenerate or complex. That is to say, both integrals \[\int_0^te^{2x_3(\tau)}\mathrm{d}\tau\;\text{and}\;\int_0^te^{-2x_1(\tau)}\mathrm{d}\tau\] will diverge to $\pm\infty$ as $t\to\pm\infty$. Hence, there exists a positive (respectively negative) time $t_c$ such that $\frac1{m_3}=0$ or $\frac1{m_1}=0$ whenever  $m_3(0)<0$ or $m_1(0)>0$ (respectively $m_3(0)>0$ or $m_1(0)<0$). This proves the claim in view of lemma \ref{lem:analm1m3}.

\textit{Case 2}: There is a double root, i.e. $\lambda_1=\lambda_2\neq\lambda_3$, then $\lambda_1$ will not be the double root of $B(z)$ (respectively $\tilde B(z)$), therefore at least one of $b_1^{(1)}(0)$ and $b_1^{(2)}(0)$ (respectively $\tilde b_1^{(1)}(0)$ and $\tilde b_1^{(2)}(0)$) is nonzero. This also implies that there exists at least one increasing and one decreasing exponential function in the expansions of both $e^{x_3}$ and $e^{-x_1}$. To show that $\lambda_1$ will not be the double root of $B(z)$, we only need to note that $\lambda_1$ must be the double root of $A(-z)C(z)$ if $\lambda_1$ is the double root for both $A(z)$ and $B(z)$ according to Lemma \ref{lem:FundIdentity}. However, $C(\lambda_1)$ must be nonzero, otherwise $\Psi(z)$ is identically equal to $0$ which leads to a contradiction. Therefore $-\lambda_1$ must be the double root of $A(z)$, which also leads to a contradiction.
\end{proof}


%
%
%
%
%
%

It is immediate from the above theorem that
\begin{corollary}
\mbox{}
\begin{itemize}
\item[(1)] For cases (v)$(-++)$ and (vii)$(--+)$ there exists a unique solution to the ODEs \eqref{eq:3peakons} for all positive $t$.

 \item[(2)] For cases (ii)$(++-)$ and (iv)$(+--)$ there exists a unique solution to the ODEs \eqref{eq:3peakons} for all negative $t$.

\end{itemize}
\end{corollary}
The global existence (in one time direction) has an interesting
impact on the spectrum of the boundary value problem.
\begin{theorem}The eigenvalues of the spectral problem of cases (ii,iv,v,vii) are real, simple, nonzero, and are equal to the inverses of asymptotic values of masses.  \end{theorem}

\begin{proof}
We give a complete proof for the case (ii)($++-$).
First, we note that \[m_3(t)=m_3(0)\exp\left(-2\int\limits_t^0 [m_1(s)e^{x_1(s)-x_3(s)}+m_2(s)e^{x_2(s)-x_3(s)}]\mathrm{d}s\right)\neq0 .\] For negative $t$,
recalling that $m_1,m_2$ will remain positive, we obtain
$$\abs {m_3(t)}\leq \abs{m_3(0)}. $$
Since $m_1\leq m_1+m_2\leq M_1+\abs{m_3(t)}$ we derive an upper bound
on $m_1$, and thus on $m_2$, namely
$$m_i(t)\leq M_1+\abs{m_3(0)}, \qquad i=1,2.$$
We observe that $0<M_1+\abs{m_3(0)}$, otherwise $m_1=m_2=0$ for all times.

\noindent {\textrm{\bf Claim I}}
$$\frac{\abs{M_3}}{(M_1+\abs{m_3(0)})\abs{m_3(0)}}<m_i(t)<M_1+\abs{m_3(0)},
\quad t<0. $$
\begin{proof} (Claim I)
We only need to prove the lower bound.  To this end we estimate:
$$
\abs{M_3}<m_1(t)m_2(t)\abs{m_3(t)}<m_1(t)m_2(t)\abs{m_3(0)},
$$
and use the upper bound above on one of the factors $m_1$ or $m_2$.
\end{proof}
Replacing in the estimate for $\abs{M_3}$ both factors $m_1$ and $m_2$ with their upper bounds we extend the claim to the bound on $m_3$.

\noindent{\textrm{\bf Claim II}}
$$\frac{\abs{M_3}}{(M_1+\abs{m_3(0)})^2}<\abs{m_3(t)}<\abs{m_3(0)}.  $$
\noindent{\textrm{\bf Claim III}}
\[x_i(t)-x_j(t)\to-\infty,\quad (1\leq i<j\leq 3), \] when $t\to -\infty$.
\begin{proof} (Claim III)
The following estimate holds:
\begin{equation*}
\begin{split}
&\frac{\abs{M_3}}{(M_1+\abs{m_3(0)})^2\abs{m_3(0)}}\leq \\
&\exp{\big(-2 \int_{-\infty}^0\big[m_1(s)e^{x_1(s)-x_3(s)}+m_2(s)e^{x_2(s)-x_3(s)}\big]\mathrm{d}s\big)}
\end{split}
\end{equation*}
and thus
$$
\int_{-\infty}^0\big[m_1(s)e^{x_1(s)-x_3(s)}+m_2(s)e^{x_2(s)-x_3(s)}\big]\mathrm{d}s< \infty,
$$
which, in view of the boundedness of $m_1$ and $m_2$, implies \[\int\limits_{-\infty}^0e^{x_i(s)-x_3(s)}\mathrm{d}s<+\infty,\quad i=1,2. \]
In addition, direct estimates on equations \eqref{eq:3peakons} using the bounds on $m_1,m_2,m_3$, show that the velocities are bounded, which means the derivative of the integrand $e^{x_i(s)-x_3(s)}\,(i=1,2)$ is bounded. Therefore \[\lim_{s\to-\infty}e^{x_i(s)-x_3(s)}=0,\quad (i=1,2),\] which is equivalent to \[\lim_{t\to-\infty}x_i(t)-x_3(t)=-\infty,\quad (i=1,2). \]

Now we turn to the ODE for $m_1$ \[m_1(t)=m_1(0)\exp\left(-2\int\limits_t^0 [m_2(s)e^{x_1(s)-x_2(s)}+m_3(s)e^{x_1(s)-x_3(s)}]\mathrm{d}s\right)>0. \]
Since $m_1$ is, for negative $t$, bounded from above and from below away
from $0$, the integral \[\int\limits_{-\infty} ^0 m_2(s)e^{x_1(s)-x_2(s)}\mathrm{d}s<+\infty, \]
and $\int\limits_{-\infty} ^0 e^{x_1(s)-x_2(s)}\mathrm{d}s<+\infty$.
Repeating {\it verbatim} the arguments from the previous case, we get \[\lim_{t\to-\infty}x_1(t)-x_2(t)=-\infty.  \quad \]
\end{proof}

Now, since the improper integrals appearing in the formulas for $m_1$ and $m_3$ exist, we can take the limit $t\to -\infty$.  Let us denote those
  limits by $m_1(-\infty)$, $m_3(-\infty)$ respectively.  Using $M_1$ we conclude
  that $m_2$ also has a limit, say, $m_2(-\infty)$.  The characteristic polynomial
  $A(z)$ (see \eqref{eq:charpoly}) reads:
  \begin{equation*}
  \begin{split}
  &A(z)=(1-\frac{z}{\lambda_1})(1-\frac{z}{\lambda_2})(1-\frac{z}{\lambda_3})\\
  &=(1-m_1(-\infty)z)(1-m_2(-\infty)z)(1-m_3(-\infty)z).
  \end{split}
  \end{equation*}

 \noindent {\textrm{\bf Claim IV}}
 $\lim_{t\to -\infty}x_1=\lim_{t\to -\infty}x_2=-\lim_{t\to -\infty}x_3=-\infty.$
 \begin{proof} (Claim IV)
 We prove first the claim for $x_1$.  The right hand side in the equation for $\dot x_1$ (see \eqref{eq:3peakons}) reads $m_1+m_2e^{x_1-x_2}+m_3e^{x_1-x_3}$.
 When $t\to -\infty$ the second and third terms go to $0$.  The limit $m_1(-\infty)>0$ because of the lower bound on $m_1$.  Hence there exists a constant $
 \alpha >0$ and another constant $t^*<0$ such that
 \begin{equation*}
0<\alpha <\dot x_1, \qquad \text{for all } t\leq t^*.
\end{equation*}
 Integrating this inequality from (negative) $t$ to $t^*$ we obtain
 $x_1(t)\leq \alpha t +C$, where $C$ is a constant, which proves that $x_1\to -\infty $
 as $t\to -\infty$. The same argument works for $x_2$.  Since $m_3(-\infty)<0$
 we get an opposite estimate for $x_3$, namely, there exists a constant $\beta<0$
 and another constant $t^{**}<0$ such that
 \begin{equation*}
\dot x_3<\beta<0, \qquad \text{for all } t\leq t^{**}
\end{equation*}
which, upon integration, yields
$\beta t+D \leq x_3(t)$, forcing $x_3(t) \to +\infty$ when $t\to -\infty$.
 \end{proof}
  Since $m_3(-\infty)<0$, to prove simplicity, we need to show that
  the remaining two positive limits $m_1(-\infty)$ and $m_2(-\infty)$ are
  distinct.
  We adapt the proof of a similar statement in \cite{ls-cubicstring}.
  First, we observe that in view of \eqref{eq:3peakons} $\dot m_2>0$, hence
  $m_2(t)$ is increasing.  The right hand side of the equation for $\dot m_1$
  reads (after dividing by $2m_1$) $-m_2e^{x_1-x_2}-m_3e^{x_1-x_3}=
  e^{x_1-x_2}(-m_2-m_3 e^{x_2-x_3})$.  Since the second term goes to $0$ as
  $t\to -\infty$ , and $m_1$ is bounded away from $0$, we obtain that there exists $t^*<0$ such that:
  \begin{equation*}
\dot m_1<0, \qquad \text{ for all } t<t^*,
\end{equation*}
So $m_1$ is decreasing for $t$ sufficiently large and negative.
Suppose, to derive a contradiction, that $m_1(-\infty)=m_2(-\infty)$ then
\begin{equation*}
m_1(t)-m_2(t)<0, \quad \text{ for all } t<t^*.
\end{equation*}
On the other hand,
\begin{equation*}
\int_t^0 (\dot x_1-\dot x_2)\, d\, \tau=x_1(0)-x_2(0)-(x_1(t)-x_2(t))=\int_t^0(m_1-m_2)d\, \tau+B(t),
\end{equation*}
where $B(t)$ has a finite limit as $t\to -\infty$.  The left hand side
diverges to $+\infty$ as $t\to -\infty$, while the right hand side can only
diverge to $-\infty$ based on the inequality above.  This contradiction shows that $m_1(-\infty)>m_2(-\infty)$.

The proof for the other three cases is analogous.
\end{proof}

\begin{remark}
It is helpful to have an intuitive understanding of the above theorems.
The emerging picture is this: if one has a swarm of peakons colliding
with a swarm of antipeakons then the system essentially behaves as if it
were
a peakon-antipeakon pair.  Thus the system has the following characteristics:
\begin{enumerate}
\item peakons and antipeakons are asympotically (in an appropriate time direction) free with asymptotic velocities $\dot x_j=\frac{1}{\lambda_j}$ which are distinct
\item peakons, antipeakons separate, that is $x_i-x_j \to -\infty$ for $i<j$.
\end{enumerate}
\end{remark}
There are only two cases left: $(-+-)$ and $(+-+)$.  In view of the reflection symmetry
(see Lemma \ref{lem:symmetry}), it suffices to analyze only one of them.  We
choose $(-+-)$.
\begin{theorem} \label{thm:confinement} Suppose the mass signature is $(-+-)$, then there exists a positive time $t_c$ such that $x_1(t_c)<x_2(t_c)=x_3(t_c)$ and a negative time $t_c^*$
such that $x_1(t_c^*)=x_2(t_c^*)<x_3(t_c^*)$. \end{theorem}

\begin{proof}
We only need to prove that $m_1,m_2$ will never collide at a positive time. If not,  assume that there exists a positive $t_c$ such that $m_1\to-\infty,m_2\to+\infty$ and $m_3$ remains bounded when $t\to t_c$. Then by equation \eqref{eq:3peakons} \[\frac{\mathrm{d}}{\mathrm{d}t}(x_1-x_2)=(1-e^{x_1-x_2})(m_1-m_2-m_3e^{x_2-x_3})<0, \] when $t$ is sufficiently close to $t_c$. This contradicts $x_1-x_2\to 0$.  Likewise, for negative times, we need
to eliminate a collision of $x_2$ and $x_3$.  However,
\begin{equation*}
\frac{d}{dt}(x_2-x_3)=(1-e^{x_2-x_3})(m_2-m_3)-m_1(e^{x_1-x_3}-e^{x_1-x_2})>0
\end{equation*}
for sufficiently close to the collision time $t_c^*<0$, again contradicting
that $x_2-x_3\to 0$ as $t \to t_c^*+$.  \end{proof}
\noindent Below we put the results of our investigation in a form of Table \ref{table:1}.  Speaking of the asymptotic behaviour we denote
by AF the system which is asymptotically free in both time directions, if it is
only in one, say in the direction of positive time, then we abbreviate it as AF+, etc.

\begin{table}[ht]
\caption{Correlations of masses, spectrum and asymptotics} 
\centering 
\begin{tabular}{c c c c} 
\hline\hline 
Mass signature & Spectrum& Asymptotic behaviour & Collisions\\ [0.5ex] 
\hline 
$+++$& $+++$&\small{AF}& none \\ 
$++-$& $++-$& \small{AF-}& $0<t_c$\\
$+-+$& $\lambda_1<0
<\textrm{Re}\,\lambda_2\leq \textrm{Re}\, \lambda_3 $ & \small{confined}& $t_c^*<0<t_c$\\
$+--$& $--+$&\small{AF-}&$0<t_c$ \\
$-++$& $-++$& \small{AF+}& $t_c^*<0$\\
$-+-$&$\textrm{Re} \,\lambda_1\leq \textrm{Re} \,\lambda_2<0<\lambda_3$&\small{confined}&$t_c^*<0<t_c$\\
$--+$&$--+$&\small{AF+}&$t_c^*<0$\\
$---$&$---$&\small{AF}&none\\
[1ex] 
\hline 
\end{tabular}
\label{table:1} 
\end{table}

We would like to conclude this section by discussing briefly the question of spectral data and in which sense the formulas
obtained in \cite{ls-cubicstring} can be used to produce multipeakon solutions.
We recall, in the notation of that paper,
\begin{equation}
    \label{eq:n-peakon-solution}
    x_{k'} = \log \frac{U_k}{V_{k-1}},
    \qquad
    m_{k'} =
    \frac{(U_k)^2 \, (V_{k-1})^2}{W_k W_{k-1}}
    \qquad
    (k=1,\dots,n),
  \end{equation}
  where $k'=n+1-k$, and $U_k$ and $V_k$ are certain rational functions of
  spectral data
  \begin{equation}
  \label{eq:positive-spectraldata}
  \mathcal{R} =
  \bigl\{ (\lambda,b)\in\R^{2n} \;:\; 0<\lambda_1<\cdots<\lambda_n, \; \text{all $b_i>0$} \bigr\}.
\end{equation}
  (see formulas (2.44) and (2.45) in \cite{ls-cubicstring} for definitions), while $W_j=U_jV_j-U_{j+1}V_{j-1}$.
  Clearly, for multipeakons,
  \begin{enumerate}
 \item the spectrum is no longer positive, simple, or even real,
 \item the residues $b_j$ can be negative and, in general, complex,
 \item the anti-resonance condition $\lambda_i+\lambda_j=0$
 renders the formulas not directly applicable
 \item once cannot extend the formulas beyond a collision point
 because there is no longer guarantee that $x_i<x_j$ for $i<j$.
 \end{enumerate}
\begin{figure}[heretp]
\begin{center}
\subfigure[DP mass signature $(++-)$; $m_2$ collides with $m_3$ at $t_c=0$.
Crossing of $x_2$ with $x_3$]{\label{fig:edge-a}\includegraphics*[viewport=60 60 600 600, width=8cm,height=8cm, scale=0.75]{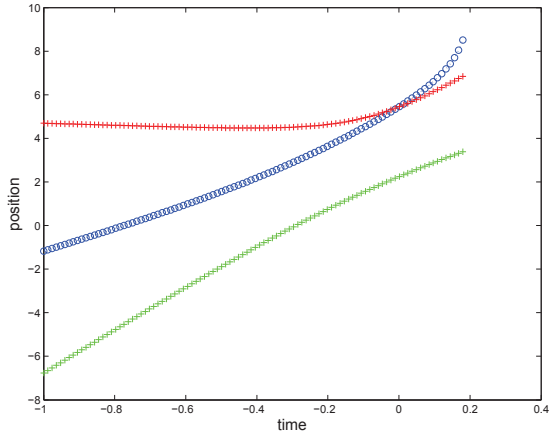}}\\
\subfigure[CH mass signature $(++-)$; $m_2$ collides with $m_3$ at $t_c=0$.
Collision of $m_2$ with $m_3$, followed by another collision of $m_2$ with $m_1$.  No crossing.  ]{\label{fig:edge-b}\includegraphics*[viewport=60 60 600 600, width=8cm,height=8cm]{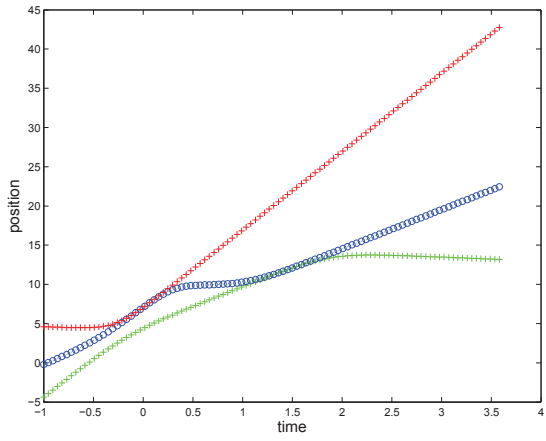}}
\end{center}
\caption{Comparison of DP and CH collisions}
\label{fig:2}
\end{figure}
\begin{example} Figure \ref{fig:edge-a} on page \pageref{fig:edge-a} illustrates how the formulas
would work for the mass signature $(++-)$.  A point from the spectral set  is chosen so that the
collision occurs at $t_c=0$.  The continuation beyond the collision point
would force new ordering $x_1<x_3<x_2$ which means that this is not the original multipeakon problem  given by equation \eqref{eq:b-peakon-ode} for $b=3, n=3$, even though
the solution still satisfies equation \eqref{eq:3peakons}, albeit in the wrong region.
This should be contrasted
with the behaviour of peakons at collision points for the CH equation as
illustrated by figure \ref{fig:edge-b} on page \pageref{fig:edge-b}.  The second particle
bounces between $m_3$ and $m_1$ and, consequently, no change of ordering is required.

\end{example}

\begin{example}
In this example we consider the mass signature $(+-+)$.
 Figure \ref{fig:4} on page \pageref{fig:4} illustrates how the formulas
\eqref{eq:n-peakon-solution} would work for the mass signature $(+-+)$.
In particular, as predicted by Theorem \ref{thm:confinement} (using the reflection
symmetry), there are two collision points, one for negative time, one for positive.
\begin{center}
\begin{figure}[here]
\includegraphics*[viewport=50 60 600 600, width=8cm,height=8cm]{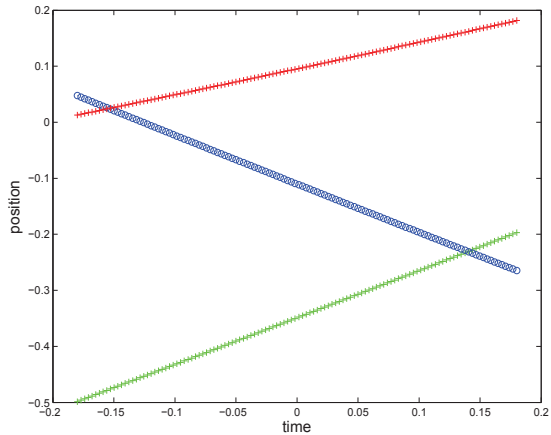}
\caption{DP mass signature $(+-+)$; $m_2$ collides with $m_3$ at $t_c^*<0$
while $m_2$ collides with $m_1$ for $t_c>0$.  Confined state.  }
\label{fig:4}
\end{figure}

\end{center}
\end{example}

\section{Acknowledgments}
The authors would like to thank H. Lundmark for numerous perceptive comments and suggestions for improvements.

This work was supported by National Natural Science
Funds of China \newline
 [NSFC10971155 to L.Z]; and National Research Science and Engineering Council of Canada [NSERC163953 to J.S].  Both authors would like to thank the Department of Mathematics and Statistics of the University of Saskatchewan for making the collaboration possible.

\end{document}